\documentclass{IEEEtran}
 \usepackage{cite}
\usepackage{psfrag}
\usepackage{latexsym, amsmath, subfigure,color, amsfonts, amssymb,graphicx}
\usepackage{algorithm}
\usepackage{algorithmic}
\usepackage{epstopdf}

\newtheorem{theorem}{Theorem}[section]

\newtheorem{corollary}[theorem]{Corollary}

\newtheorem{definition}[theorem]{Definition}

\begin{document}

\title{Differential Private Noise Adding Mechanism \\and Its Application on Consensus$^*$}
\author{Jianping He and Lin Cai
\thanks{$^*$ Part of the preliminary result of this work was accepted by IEEE Conference
on American Control Conference (ACC), 2017 \cite{he17acc}.}
\thanks{The Dept. of Electrical \& Computer Engineering at the University of Victoria, BC, Canada {\tt\small{jphe@uvic.ca}}; {\tt\small cai@ece.uvic.ca}}
 }

\maketitle
\begin{abstract}
Differential privacy is a formal mathematical {stand-ard} for quantifying the degree of that individual privacy in a statistical database is preserved. To guarantee differential privacy, a typical method is adding random noise to the original data for data release. In this paper, we investigate the conditions of differential privacy considering the general random noise adding mechanism, and then apply the obtained results for privacy analysis of the privacy-preserving consensus algorithm. Specifically, we obtain a necessary and sufficient condition of $\epsilon$-differential privacy, and the sufficient conditions of $(\epsilon, \delta)$-differential privacy.  We apply them to analyze various random noises. For the special cases with known results, our theory matches with the literature; for other cases that are unknown, our approach provides a simple and effective tool for differential privacy analysis.  Applying the obtained  theory, on privacy-preserving consensus algorithms, it is proved that the average consensus and $\epsilon$-differential privacy cannot be guaranteed simultaneously by any privacy-preserving consensus algorithm.
\end{abstract}


\IEEEpeerreviewmaketitle

\section{{Introduction}} \label{Intro}

Differential privacy, a famous and widely used privacy concept, aims to minimize the chances of identifying a single record in a release of a large database \cite{Dwork06}. Differential privacy means that the presence or absence of any individual record in the database will not affect the statistics significantly \cite{TIT16}. Thus, the adversary has a low chance to identify the individual's record with the released information and any auxiliary information under differential privacy. Differential privacy has been a formal framework to quantify the degree of that the individual privacy  is preserved while releasing useful statistical information about the database \cite{TIT16}. We refer the readers to  \cite{Dwork08, Dwork2014} by Dwork et al. for the detailed introduction of differential privacy, including the motivation, background, the important developments of its theories and applications. More recently, Cortes et al.  \cite{Cortes16}  introduced a system and control
perspective on the topic of privacy-preserving data analysis, showing the importance of differential privacy in network and control area.


There are two kinds of differential privacy concepts which have been widely investigated in the literature. The first is
 $\epsilon$-differential privacy. The parameter $\epsilon$ expresses the privacy cost, and a smaller value of $\epsilon$ can guarantee a stronger privacy. Based on $\epsilon$-differential privacy,  an adversary cannot gain significant information about the data function of any individual agent based on the observation of the data output.  The typical approach to preserving $\epsilon$-differential privacy is adding
random noise to original data for information release. Recently, Geng and Viswanath \cite{TIT16} showed that when the adding random noise follows a staircase distribution, the optimal $\epsilon$-differential privacy can be guaranteed in terms of a general utility-based framework formulated by the authors. The second is $(\epsilon, \delta)$-differential privacy, which is a relaxed notion of privacy. In this privacy definition, the  parameter $\epsilon$ represents the privacy cost  and $\delta$ represents the probability of violating the privacy cost. For both parameters, smaller values correspond to higher privacy \cite{Barthe16}. To ensure $(\epsilon, \delta)$-differential privacy, a often-used approach is adding Gaussian noise to the pure data value for query output.

Although random noise adding mechanism has been widely-used, how to design and analyze the effectiveness of various types of noises remains a challenge. Existing work mostly focused on a few well-known noise distributions (e.g., Laplacian and  Gaussian). To fill this gap, in this paper, we first investigate the basic conditions for the random noise adding mechanism, under which differential privacy can be guaranteed. We then obtain the conditions to determine whether the differential privacy is guaranteed by the noise adding mechanism. To show this statement, we analyze the well-known noise adding mechanisms, e.g., Laplacian and  Gaussian. For the special cases with known results, our theory matches with the literature; for other cases that are unknown, our approach provides a simple and effective tool for differential privacy analysis.  In addition, we apply the theory to analyze the privacy of the privacy-preserving consensus algorithms,  a hot topic in the control and optimization area recently.  The main contributions of this paper are summarized as follows\@.
\begin{enumerate}
\item We investigate the conditions of a general random noise adding mechanism to guarantee differential privacy. We obtain a necessary and sufficient condition of $\epsilon$-differential privacy, and the sufficient conditions of $(\epsilon, \delta)$-differential privacy. Meanwhile, we provide the estimation of the upper bounds of $\epsilon$ and $\sigma$, respectively.
\item  We show that the obtained theory provides an efficient and simple approach for the analysis of  both $\epsilon$-differential privacy and $(\epsilon, \delta)$-differential privacy. Specifically, using the obtained results, we prove that Laplacian noise is $\epsilon$-differential private, and Gaussian noise is $(\epsilon, \delta)$-differential private, and uniform noise is $(0, \delta)$-differential private,  consistent with the claims in the existing works.
\item   We apply the theorems of differential privacy to analyze the privacy of general privacy-preserving consensus algorithm. We obtain the necessary  condition and the sufficient condition for the algorithm under which differential privacy is achieved. We prove that it is impossible to achieve the average consensus and $\epsilon$-differential privacy simultaneously.
\end{enumerate}
Different from the existing work, the proposed results can be used to the random  noise adding mechanism with any noise distributions.

The remainder of this paper is organized as follows\@.  Section \ref{sec:pre} formulates the problem.  In Section \ref{sec:mainresult},  we provide the basic theoretical results of differential privacy\@.   Section \ref{sec:ap} studies the application on privacy-preserving consensus algorithm. The related works are given in Section \ref{sec:rw}. Conclusions are summarized in Section \ref{sec:conclusions}\@.

\section{Preliminary and Problem Formulation}\label{sec:pre}

\subsection{Preliminary}
Let $V=\{1, 2,...,n\}$ be the set of nodes (users). Following \cite{huang12, huang15, Nozari16, Han2017},  adjacency and differential privacy, respectively, are defined as follows.
\begin{definition} [$\sigma$-adjacency]
Given $\sigma \in \mathbf{R}^+$, the state vector $x$ and $y$ are $\sigma$-adjacent if, for some $i_0\in V$,
\begin{equation} \label{adjacent}
|x_i-y_i|\leq \left\{ \begin{aligned}
        & \sigma,  &&  \textrm{if} ~ i=i_0; \\
       &0, &&   \textrm{if} ~ i\neq i_0,
                          \end{aligned} \right.
                          \end{equation}
for $i \in V$, where $x, y \in \mathbf{R}^n$.
 \end{definition}

From the above definition, it follows that a pair of $\sigma$-adjacent vectors $x$ and $y$ have at most one different element, and  the difference is no more than $\sigma$. For example, $x=[0, 1]$ and $y=[1, 1]$ are $1$-adjacent vectors.

\begin{definition} [$(\epsilon, \delta)$-differential privacy]
A randomized mechanism $\mathcal {A}$ with domain $\Omega$ is $(\epsilon, \delta)$-differentially private if, for any pair $x$ and $y$ ($x, y\in \Omega$) of $\sigma$-adjacent state vector and any set $\mathcal {O}\subseteq \textrm{Range}(\mathcal {A})$, where $\textrm{Range}(\mathcal {A})$ is the domain of the output under mechanism $\mathcal {A}$,
 \begin{equation}\label{diffprivacy}
\Pr\{\mathcal {A}(x)\in \mathcal {O}\}\leq e^{\epsilon}\Pr\{\mathcal {A}(y)\in \mathcal {O}\}+\delta.
\end{equation}
 \end{definition}

In the above privacy definition, there are two key parameters, $\epsilon$ and $\delta$,  which represent the privacy cost and the probability of violating the privacy cost, respectively. For both of these parameters, smaller values offer stronger privacy guarantees. If $\delta=0$, we say that $\mathcal {A}$ is $\epsilon$-differentially private, which provides a stronger privacy than $(\epsilon, \delta)$-differential privacy.

 Table~\ref{table:definitions} summarizes a few important notations in this paper for easy reference.

\subsection{Problem Formulation}
\textbf{General Random Mechanism}: We consider a general random noise adding mechanism. Assume that the randomized mechanism $\mathcal {A}: \Omega \rightarrow\textrm{Range}(\mathcal {A})$ satisfies
\begin{equation}\label{randoma0}
\mathcal {A}(x)=x+\theta, ~\forall ~x\in \Omega,
\end{equation}
where $\theta\in \Theta$ is a random noise vector with $f_{\theta_i}(z)$ as the PDF of its $i$-th element $\theta_i$, and $\Theta\subseteq R^n$.  We assume that each $f_{\theta_i}(z)$ is a continuous or sectionally continuous function, and $\theta_i$ and $\theta_j$ are independent from each other for $\forall i\neq j$. Then, we have $\textrm{Range}(\mathcal {A})=\Omega \oplus \Theta$, where $\oplus$ denotes the Minkowski sum between two set, i.e., any element in $\textrm{Range}(\mathcal {A})$ will equal to the sum of two elements respectively in sets $\Omega$ and $\Theta$. The random mechanism $\mathcal {A}$ defined in (\ref{randoma0}) is a general noise adding mechanism, where $x$ could be substituted by a general invertible function of $x$ with  Lipschitz condition and $\theta$ could also be a function of  random variables\footnote{ For example, we consider a more general mechanism as follows
\[\mathcal {A}(x)=g(x)+h(\theta), ~\forall ~x\in \Omega, \theta\in \Theta,\]
where $g(x)$ is a function of $x$ satisfying $|g(x)-g(y)|\leq L |x-y|$ (where $L$ is a Lipschitz constant) and $g(x)\neq g(y)$ when $x\neq y$ and $h(\theta)$ is a function of $\theta$. We can use the similar analytical approach given in this paper to analyze differential privacy of the above mechanism.}.

\textbf{Problem of Interests}:  The goal of this paper is to investigate the following two questions: i) what are the sufficient and necessary conditions of differential privacy considering the randomized mechanism $\mathcal {A}$ defined in (\ref{randoma0}); and ii) how to extend and apply the obtained results for privacy analysis on the privacy-preserving consensus algorithm, an important distributed iteration algorithm in the control area. We will solve these two problems in the following two sections, respectively.

  \begin{table}[t] \tabcolsep 1pt \caption{Important Notations} \vspace*{4pt} \centering \tabcolsep 0.5mm
    \begin{tabular}{c||l}
    \hline Symbol  &                   Definition\\ \hline
    $x, y$                 & a pair of $\sigma$-adjacent $n$ dimensional  vector\\
  $\mathcal {A}$ & a random mechanism\\
  $\sigma$ & a parameter expressing the adjacency between vectors\\
  $\epsilon$ & a parameter expressing the privacy cost\\
   $\delta$ & a parameter expressing the probability of the violating the privacy\\
  $\mathcal {O}$ & a subset of $\textrm{Range}(\mathcal {A})$ \\
    $\mathcal {O}_i$ & the set of $i$-th column element of  $\mathcal {O}$\\
  $\textrm{Range}(\mathcal {A})$ & the domain of the output under mechanism $\mathcal {A}$\\
    $ f_{\theta_i}(z)$ & the probability density function of random variable $\theta_i$ \\
    $\mu$                 & the function of Lebesgue measure\\
    $\Phi_i^0$   & the zero point set of the function $f_{\theta_i}(z)$ \\
    $W$                 & the weight matrix in a consensus alogrithm\\
        $\bar{x}$                 & the average value of all node states\\
        \hline
    \end{tabular}
\vspace*{-4pt}
\label{table:definitions}
\end{table}

\section{Conditions of Differential Privacy} \label{sec:mainresult}
In this section, first, the basic conditions of  differential privacy considering $\mathcal {A}$ defined in (\ref{randoma0}) are obtained, followed by the estimations of the privacy parameters. Then, we show that the obtained conditions provide efficient criteria of differential privacy through case studies, where the Laplacian, Gaussian, and Uniform noises are investigated, respectively, by using the developed theoretical results.

\subsection{Necessary and Sufficient Condition}
In this subsection, a necessary and sufficient condition of $\epsilon$-differentially private is given in the following theorem.
\begin{theorem}\label{theorem0:nsc}
 $\mathcal {A}$ is $\epsilon$-differentially private if and only if (iff) the following two conditions hold,
\begin{enumerate}
\item[c1:] zero measure of the zero-point set, i.e.,
\begin{equation}\label{cedu0}
\mu\left(\cup_{i=1}^n \Phi_i^0\right)=0
\end{equation}
where $\Phi_i^0= \{z| f_{\theta_i}(z)=0, z\in \mathbf{R}\}$ is the zero point set and $\mu(\cdot)$ is the  Lebesgue measure;
\item[c2:] there exists a positive constant $c_b$ such that
\begin{equation}\label{factpingyi0}
\sup_{\hat{\sigma}\in[-\sigma, \sigma], f_{\theta_i}(z)\neq 0} {f_{\theta_i+\hat{\sigma}}(z)\over f_{\theta_i}(z)} \leq c_b,
\end{equation}
for $\forall i\in V$, where the acnodes space (isolated point) of both $f_{\theta_i+\hat{\sigma}}$ and $f_{\theta_i}(z)$ are not considered.
\end{enumerate}
\begin{proof}
$\Leftarrow$: We prove the necessity by contradiction.

First, we prove that  (\ref{cedu0}) is a necessary condition.  Assume that  $\mu\left(\cup_{i=1}^n \Phi_i^0\right)>0$, which means that there exists at least one $i_0\in V$ with $\mu\left(\Phi_{i_0}^0\right)>0$. Thus, there should exist an continuous interval $[a, b]$ such that
\begin{equation}\label{eq:fz0}
f_{\theta_{i_0}}(z)=0, ~\forall~ z\in [a, b],
\end{equation}
and
\begin{equation}\label{eq:fzl0}
 f_{\theta_{i_0}}(z)>0,~\textrm{for}~ z>b,
 \end{equation}
or $ f_{\theta_{i_0}}(z)>0$ for $z<a$ {\footnote{Without loss of generality, suppose $f_{\theta_i}(z)>0$ holds for $z>b$ in the following proof.}}.

For $\sigma$-adjacent state vectors $x$ and $y$, we let $x_{i_0}=y_{i_0}-\sigma$ and $x_{i}=y_{i}$ (when $i\neq i_0$),  where $\sigma<b-a$.  Then, with  (\ref{randoma0}), we have $\mathcal {A}(x_{i_0})=x_{i_0}+\theta$ and $\mathcal {A}(y_{i_0})=y_{i_0}+\theta$. Define a subset $\mathcal {O}_{i_0}=[y_{i_0}+a, y_{i_0}+b]$, which satisfies $\mathcal {O}_{i_0}\subseteq \textrm{Range}(\mathcal {A}(x_{i_0}))$. From (\ref{eq:fz0}) and  (\ref{eq:fzl0}), it follows that
\begin{align*}
& \Pr\{\mathcal {A}(x_{i_0})\in \mathcal {O}_{i_0}\}=\int_{y_{i_0}+a}^{y_{i_0}+b} f_{x_{i_0}+\theta_{i_0}}(z) \text{d} z
\\&=\int_{y_{i_0}+a-x_{i_0}}^{y_{i_0}+b-x_{i_0}} f_{\theta_{i_0}}(z) \text{d} z=\int_{a+\sigma}^{b+\sigma} f_{\theta_{i_0}}(z) \text{d} z\\&=\int_{a+\sigma}^{b} f_{\theta_{i_0}}(z) \text{d} z+\int_{b}^{b+\sigma} f_{\theta_{i_0}}(z) \text{d} z \\&=\int_{b}^{b+\sigma} f_{\theta_{i_0}}(z) \text{d} z>0,
 \end{align*}
while
 \begin{align*}
\Pr\{\mathcal {A}(y_{i_0})\in \mathcal {O}_{i_0}\}&=\int_{y_{i_0}+a}^{y_{i_0}+b} f_{y_{i_0}+\theta_{i_0}}(z) \text{d} z
\\&=\int_{a}^{b} f_{\theta_{i_0}}(z) \text{d} z=0.
 \end{align*}
Hence, it follows that
  \begin{align}\label{x0xy0y}
&{\Pr\{\mathcal {A}(x)\in \mathcal {O}\}\over \Pr\{\mathcal {A}(y)\in \mathcal {O}\}}\nonumber\\=&{ \Pr\{\mathcal {A}(x_{i_0})\in \mathcal {O}_{i_0}\}  \prod_{i=1, i\neq i_0}^n \Pr\{\mathcal {A}(x_{i}) \in \mathcal {O}_i\}\over \Pr\{\mathcal {A}(y_{i_0})\in \mathcal {O}_{i_0}\}  \prod_{i=1, i\neq i_0}^n \Pr\{\mathcal {A}(y_{i}) \in \mathcal {O}_i\}}\nonumber\\=&{ \Pr\{\mathcal {A}(x_{i_0})\in \mathcal {O}_{i_0}\} \over \Pr\{\mathcal {A}(y_{i_0})\in \mathcal {O}_{i_0}\} }= \infty,
 \end{align}
 where $\mathcal {O}_i$ is the domain of the $i$-th element in  $\mathcal {O}$.
It contradicts to the definition of $\epsilon$-differential privacy. Thus, one obtains that (\ref{cedu0}) is a necessary condition if $\mathcal {A}$ is $\epsilon$-differentially private.

 Second, we prove that  (\ref{factpingyi0}) is also a necessary condition. Suppose that
 \[\sup_{\hat{\sigma}\in[-\sigma, \sigma], f_{\theta_i}(z)\neq 0} {f_{\theta_i+\hat{\sigma}}(z)\over f_{\theta_i}(z)}=\infty.\]
 Then, for any given large constant $M$, there exist $z_0$ and $\hat{\sigma}\in [-\sigma, \sigma]$ such that $ f_{\theta_i}(z_0)\neq 0$  and \[{f_{\theta_i}(z_0+\hat{\sigma})\over f_{\theta_i}(z_0)}\geq M.\]
 Since (\ref{cedu0}) holds and the measure of the acnodes is also zero, we can ignore these two cases when we calculate the probability. We thus can assume that $f_{\theta_{i_0}}(z)$ is a continuous function in a small interval among $z_0$ and among $z_0+\hat{\sigma}$. Then, based on the definition of a continuous function, there exists a small positive constant $\varepsilon_0$ such that
 \[\max_{z\in[z_0,  z_0+\varepsilon_0]} f_{\theta_{i_0}}(z) \leq 2 f_{\theta_i}(z_0)\]
and
 \[\min_{z\in[z_0+\hat{\sigma},  z_0+\hat{\sigma}+\varepsilon_0]} f_{\theta_{i_0}}(z)\geq (M-1) f_{\theta_i}(z_0).\]

Then, we construct a pair of $\hat{\sigma}$-adjacent state vector $x$ and $y$ with $x_{i_0}=y_{i_0}-\hat{\sigma}$ and $x_{i}=y_{i}$ (when $i\neq i_0$). Define the set $\mathcal{O}_{i_0}^0=[y_{i_0}+z_0, y_{i_0}+z_0+\varepsilon_0]$, where $\varepsilon_0\leq \hat{\sigma}$.  Based on (\ref{randoma0}), we have
\begin{align*}
{\Pr\{\mathcal {A}(x_{i_0})\in \mathcal {O}_{i_0}^0\} \over \Pr\{\mathcal {A}(y_{i_0})\in \mathcal {O}_{i_0}^0\}}&={\int_{y_{i_0}+z_0}^{y_{i_0}+z_0+\varepsilon_0} f_{x_{i_0}+\theta_{i_0}}(z) \text{d} z \over \int_{y_{i_0}+z_0}^{y_{i_0}+z_0+\varepsilon_0} f_{y_{i_0}+\theta_{i_0}}(z) \text{d} z}\\
&={\int_{z_0+\hat{\sigma}}^{z_0+\hat{\sigma}+\varepsilon_0} f_{\theta_{i_0}}(z) \text{d} z \over \int_{z_0}^{z_0+\varepsilon_0} f_{\theta_{i_0}}(z) \text{d} z}\\&\geq{ (M-1) f_{\theta_i}(z_0)\varepsilon_0 \over 2 f_{\theta_i}(z_0)\varepsilon_0}\\&\geq{ (M-1)  \over 2}.
 \end{align*}
 Then, similar to (\ref{x0xy0y}), we infer that
 \begin{align*}
{\Pr\{\mathcal {A}(x)\in \mathcal {O}\}\over \Pr\{\mathcal {A}(y)\in \mathcal {O}\}}\geq{ (M-1) \over 2}.
\end{align*}
Note that $M$ could be an arbitrarily large constant, which implies that $\mathcal {A}$ is not $\epsilon$-differentially private. Hence, (\ref{cedu1}) also a necessary condition to ensure that  $\mathcal {A}$ is $\epsilon$-differentially private.

$\Rightarrow$:  Next, we prove the sufficiency.
Since (\ref{cedu0}) holds, the cases that $f_{\theta_i}(z)=0, \forall i\in V$ and the acnodes are ignored when we calculate a probability. Let $\mathcal {O}_i$ be the domain/set of $i$-th element in  $\mathcal {O}$ for $l=1, ..., n$.  Under (\ref{randoma0}), we have
\begin{align}\label{prx}
&\Pr\{\mathcal {A}(x)\in \mathcal {O}\}=\Pr\{x+\theta \in \mathcal {O}\}\nonumber
\\&=\Pr\{x_{i_0}+\theta_{i_0}\in \mathcal {O}_{i_0}\} \prod_{i=1, i\neq i_0}^n \Pr\{x_i+\theta_i \in \mathcal {O}_i\}
\end{align}
and
\begin{align}\label{pry}
&\Pr\{\mathcal {A}(y)\in \mathcal {O}\}=\Pr\{y+\theta \in \mathcal {O}\} \nonumber
\\&=\Pr\{y_{i_0}+\theta_{i_0}\in \mathcal {O}_{i_0}\} \prod_{i=1, i\neq i_0}^n \Pr\{y_i+\theta_i \in \mathcal {O}_i\}.
\end{align}
Since  $x_{i}=y_{i}, i\neq i_0$, we have
\begin{align}\label{prxy}
\prod_{i=1, i\neq i_0}^n \Pr\{x_i+\theta_i \in \mathcal {O}_i\}= \prod_{i=1, i\neq i_0}^n \Pr\{y_i+\theta_i \in \mathcal {O}_i\}.
\end{align}
Meanwhile, with the condition of $c_2$,  it follows
\begin{align}\label{nef}
&\Pr\{x_{i_0}+\theta_{i_0}\in \mathcal {O}_{i_0}\} =\oint_{\mathcal {O}_{i_0}} f_{x_{i_0}+\theta_{i_0}}(z) \text{d} z
\nonumber\\&=\oint_{\mathcal {O}_{i_0}} f_{y_{i_0}-\hat{\sigma}+\theta_{i_0}}(z) \text{d} z\leq \oint_{\mathcal {O}_{i_0}} c_b f_{y_{i_0}+\theta_{i_0}}(z) \text{d} z
\nonumber\\&=c_b  \Pr\{y_{i_0}+\theta_{i_0}\in \mathcal {O}_{i_0}\},
\end{align}
where we have used the fact of (\ref{factpingyi0}). Substituting  (\ref{prxy}) and (\ref{nef})  into (\ref{prx}) yields
\begin{align*}
\Pr\{\mathcal {A}(x)\in \mathcal {O}\}&\leq c_b \Pr\{\mathcal {A}(y)\in \mathcal {O}\}
\\&=e^{\log(c_b)}\Pr\{\mathcal {A}(y)\in \mathcal {O}\}
\end{align*}
which completes the proof.
\end{proof}
\end{theorem}

From the proof of the above theorem, it is easy to obtain the following corollary, which provides an upper bound estimation of the parameter $\epsilon$ in  $\epsilon$-differential privacy.
\begin{corollary}
If conditions $c_1$ and $c_2$ hold, $\mathcal {A}$ is $\epsilon$-differential private, where $\epsilon$ satisfies
$0 \leq \epsilon\leq \log(c_b)$.
\end{corollary}

Note that a smaller $\epsilon$ provides a stronger privacy guarantee. From the above corollary, we note that $\epsilon\rightarrow 0$ if $c_b\rightarrow 1$, i.e., a stronger privacy can be guaranteed when $c_b$ becomes smaller. Then, we consider a kind of random distribution which can guarantee any small $\epsilon$-differential private.  For any $c_b>1$, we construct a staircase-shaped PDF for each random variable used to the noise adding mechanism, such that the conditions $c_1$ and $c_2$ can be satisfied, which is given by
\begin{equation*}
f(z)=\left\{ \begin{aligned}
       & {1-\varrho \over 2a} \varrho^k,  && z\in [ka, (k+1)a];\\
        & {1-\varrho \over 2a}, &&   z\in [-a, a]; \\
       & {1-\varrho \over 2a} \varrho^k,  &&   z\in [-ka-a, -ka] ,
                          \end{aligned} \right.
                          \end{equation*}
where $\varrho\in (0, 1)$ and $k$ is a positive integer and $a$ is a positive constant. A staircase-shaped PDF is shown in Fig. \ref{sspdf}.
\begin{figure}[ht]
\begin{center}
\includegraphics[width=0.4\textwidth]{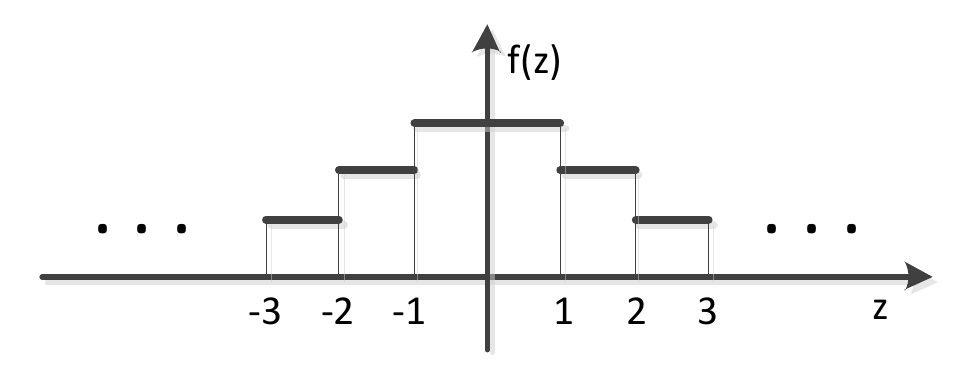}
\vspace{-5pt}
\caption{An example of the staircase-shaped PDF.} \label{sspdf}
\end{center}
\end{figure}
For the above staircase-shaped function $f(z)$, we obtain that
\[\int_{-\infty}^{+\infty}f(z) \text{d} z=(1-\varrho)+2\sum_{k=1}^{\infty}{1-\varrho \over 2} \varrho^k= 1,\]
and it thus could be a PDF function for a random variable. In this case, when $\sigma\leq 1$, it follows that
\[\sup_{\hat{\sigma}\in[-\sigma, \sigma], f_{\theta_i}(z)\neq 0} {f_{\theta_i+\hat{\sigma}}(z)\over f_{\theta_i}(z)}={1\over \varrho}.\]
Note that $\varrho$ could be any value in $(0, 1)$, which means that $c_b$ could be any value in $(1, \infty)$ by setting $\varrho={1\over c_b}$. Hence, for any small $\epsilon>0$, we can find a staircase-shaped PDF for the adding noise such that $\mathcal {A}$ is $\epsilon$-differentially private. It also should be pointed out that the staircase-shaped PDFs are proved to be the optimal distribution for noise adding mechanism in terms of utility-maximization \cite{TIT16}, where the utility is a specific cost function depending on the noise added to the query output (i.e., the different value of the adding noise will cause different cost for the query output function).

Next, we provide another necessary condition and sufficient condition of $\epsilon$-differentially private, respectively.
\begin{theorem}\label{theorem0:nsc1}
If  $\mathcal{A}$ is $\epsilon$-differentially private, then for $\forall i\in V$,   there $\nexists~ c_o\in (-\infty, +\infty)$, such that
\begin{equation}\label{cedu1}
\lim_{z\rightarrow c_0} f_{\theta_{i}}(z) =0.
\end{equation}
\end{theorem}
\begin{proof}
Suppose that there exists a bounded constant $c_0\in (-\infty, +\infty)$, such that
 $\lim_{z\rightarrow c_0} f_{\theta_{i_0}}(z) =0$.
 Since (\ref{cedu0}) holds, we can let $f_{\theta_{i_0}}(c_0) =0$ and suppose that $f_{\theta_{i_0}}(z)$ is a continuous function in a small interval among $c_0$. Then, there exists an interval  $[c_0, c_1]$ and a small $\hat{\sigma} \leq {c_1-c_0\over 2}$ such that
 \[\max_{z\in[c_0, c_0+\hat{\sigma}]} f_{\theta_{i_0}}(z)\leq \hat{\epsilon}(\hat{\sigma})\]
and
 \[\max_{z\in[c_0+\hat{\sigma}, c_1]} f_{\theta_{i_0}}(z)> \hat{\epsilon}(\hat{\sigma})\]
 where $\hat{\epsilon}(\hat{\sigma})$ satisfies $\lim_{\hat{\sigma}\rightarrow 0} \hat{\epsilon}(\hat{\sigma})=0$.
Then, we construct a pair of $\hat{\sigma}$-adjacent state vector $x$ and $y$ with $x_{i_0}=y_{i_0}-\hat{\sigma}$ and $x_{i}=y_{i}$ (when $i\neq i_0$). Define the set $\mathcal{O}_{i_0}^k=[y_{i_0}+c_0, y_{i_0}+c_0+\hat{\sigma}(k)]$, where $\hat{\sigma}(k)\leq \hat{\sigma}$.  Based on (\ref{randoma0}), we have
\begin{align*}
{\Pr\{\mathcal {A}(x_{i_0})\in \mathcal {O}_{i_0}^k\} \over \Pr\{\mathcal {A}(y_{i_0})\in \mathcal {O}_{i_0}^k\}}&={\int_{y_{i_0}+c_0}^{y_{i_0}+c_0+\hat{\sigma}(k)} f_{x_{i_0}+\theta_{i_0}}(z) \text{d} z \over \int_{y_{i_0}+c_0}^{y_{i_0}+c_0+\hat{\sigma}(k)} f_{y_{i_0}+\theta_{i_0}}(z) \text{d} z}\\
&={\int_{c_0+\hat{\sigma}}^{c_0+\hat{\sigma}+\hat{\sigma}(k)} f_{\theta_{i_0}}(z) \text{d} z \over \int_{c_0}^{c_0+\hat{\sigma}(k)} f_{\theta_{i_0}}(z) \text{d} z}\\&\geq{\hat{\epsilon}(\hat{\sigma}) \hat{\sigma}(k) \over \hat{\epsilon}(\hat{\sigma}(k))\hat{\sigma}(k)}\geq{\hat{\epsilon}(\hat{\sigma}) \over \hat{\epsilon}(\hat{\sigma}(k))}.
 \end{align*}
Let $\hat{\sigma}(k)\rightarrow 0$, one obtains
\begin{align*}
&\lim_{\hat{\sigma}(k)\rightarrow 0}{\Pr\{\mathcal {A}(x_{i_0})\in \mathcal {O}_{i_0}^k\} \over \Pr\{\mathcal {A}(y_{i_0})\in \mathcal {O}_{i_0}^k\}} \geq \lim_{\hat{\sigma}(k)\rightarrow 0}{\hat{\epsilon}(\hat{\sigma}) \over \hat{\epsilon}(\hat{\sigma}(k))}=+\infty,
 \end{align*}
which implies that $\mathcal {A}$ is not $\epsilon$-differentially private. It leads to a contradiction. Thus,  (\ref{cedu1}) is a necessary condition when $\mathcal {A}$ is $\epsilon$-differentially private, which completes the proof.
\end{proof}

(\ref{cedu1}) is actually a necessary condition of  $c_2$, which can be easily proved by contradiction. This further explains why (\ref{cedu1})  is necessary to $\epsilon$-differential privacy.

\begin{theorem}\label{theorem0:asc}
 $\mathcal{A}$ is $\epsilon$-differentially private, if, for $\forall i\in V$, there exists a positive constant $c_b$ such that
 \begin{equation}\label{factpingyi0ad}
\sup_{\hat{\sigma}\in[-\sigma, \sigma]} {f_{\theta_i+\hat{\sigma}}(z)\over f_{\theta_i}(z)} \leq c_b.
\end{equation}
\end{theorem}
\begin{proof}
Note that if (\ref{factpingyi0ad}) can guarantee  both conditions $c_1$ and $c_2$ to be satisfied, then this theorem can be  proved from Theorem \ref{theorem0:nsc}.

First, we prove that (\ref{factpingyi0ad})  guarantees that condition $c_2$ is satisfied. By comparing (\ref{factpingyi0})  and  (\ref{factpingyi0ad}), we note that the constraint $f_{\theta_i}(z)\neq 0$ in  (\ref{factpingyi0}) is removed in  (\ref{factpingyi0ad}), which means that (\ref{factpingyi0ad}) provides a more general result than (\ref{factpingyi0}). Hence, one infers that condition $c_2$ is guaranteed by (\ref{factpingyi0ad}) directly.

Then, we prove that  (\ref{factpingyi0ad}) can also guarantee condition $c_1$. First, suppose that $c_1$ is not true,  then there exists a continuous interval such that $f_{\theta_i}(z)=0$ for $z$ in this interval. Second, since $f_{\theta_i}(z)$ is a PDF of a random variable, we have $f_{\theta_i}(z)\geq 0$ and $\int_{-\infty}^{\infty}f_{\theta_i}(z) \text{d} z=1$. Thus, there exists a continuous interval such that $f_{\theta_i}(z)>0$ holds in this interval. Then, we further infer that there exist two continuous intervals $(a, b)$ and $(b, c)$ such that $f_{\theta_i}(z)=0$ for $z\in(a, b)$ and $f_{\theta_i}(z)>0$ for $z\in(b, c)$.  It means that
\begin{equation}\label{factpingyi0add}
\sup_{\hat{\sigma}\in[-\sigma, \sigma]} {f_{\theta_i+\hat{\sigma}}(z)\over f_{\theta_i}(z)}\geq \sup_{\hat{\sigma}\in[-\sigma, \sigma]}{f_{\theta_i}(b+{\hat{\sigma}\over 2})\over f_{\theta_i}(b-{\hat{\sigma}\over 2})}=\infty,
\end{equation}
which leads to a contradiction. Therefore, we have that $c_1$ is also true under (\ref{factpingyi0ad}).

We thus have completed the proof.
\end{proof}

From the above proof, we know that  (\ref{factpingyi0ad})  is a stronger condition than  conditions $c_1$ and $c_2$. Thus, (\ref{factpingyi0ad})  is  a sufficient but not necessary condition.

\subsection{Sufficient Condition for $(\epsilon, \delta)$-Differential Privacy}
In this subsection, we study the relaxed differential privacy, named $(\epsilon, \delta)$-differential privacy.  We obtain the sufficient conditions to guarantee that $\mathcal {A}$ provides $(\epsilon, \delta)$-differential privacy, followed by the estimations of both the parameters $\epsilon$ and $\delta$.
\begin{theorem}\label{theorem1:scp}
If (\ref{factpingyi0}) holds, then $\mathcal {A}$ is $(\epsilon, \delta)$-differentially private, where $\epsilon$ and $\delta$ satisfy $ \epsilon\leq \log(c_b)$ and
\begin{equation}\label{adpc1}
\delta\leq\max_{i\in V} \oint_{\Phi_{i}^0} f_{\theta_{i}}(z+\sigma)  \text{d} z,
\end{equation}
respectively, and if (\ref{cedu0}) holds, we have $\delta=0$, i.e., $\mathcal {A}$ is $\epsilon$-differentially private.
\begin{proof}
Similarly, assume the $\sigma$-Adjacency state vectors $x$ and $y$ satisfy $y_{i_0}=x_{i_0}+ \sigma$  and $x_{i}=y_{i}, i\neq i_0$, and define $\mathcal {O}_l$ to be the $l$-th column element set of  $\mathcal {O}$  for $l=1, ..., n$. Then,  we have that (\ref{prx}), (\ref{pry}) and (\ref{prxy}) still hold true.

First, we consider the case that $\mu(\cup_{i\in V}\Phi_i^0)>0$.  Given the non-zero measure of the zero point set, (\ref{nef}) no longer holds true but we can obtain the following result,
\begin{align}\label{prxyz}
&\Pr\{x_{i_0}+\theta_{i_0}\in \mathcal {O}_{i_0}\} =\oint_{\mathcal {O}_{i_0}} f_{x_{i_0}+\theta_{i_0}}(z) \text{d} z
\nonumber \\&=\oint_{\mathcal {O}_{i_0}} f_{y_{i_0}-\sigma+\theta_{i_0}}(z) \text{d} z
\leq \oint_{\mathcal {O}_{i_0}} c_b f_{y_{i_0}+\theta_{i_0}}(z) \text{d} z\nonumber\\&+\oint_{\{\Phi_{i_0}^0+y_{i_0}\}\cap \mathcal {O}_{i_0}} f_{y_{i_0}-\sigma+\theta_{i_0}}(z) \text{d} z
\nonumber\\&\leq c_b  \Pr\{y_{i_0}+\theta_{i_0}\in \mathcal {O}_{i_0}\}+\oint_{\{\Phi_{i_0}^0+y_{i_0}\}} f_{y_{i_0}-\sigma+\theta_{i_0}}(z) \text{d} z
\nonumber \\&\leq c_b  \Pr\{y_{i_0}+\theta_{i_0}\in \mathcal {O}_{i_0}\}+\oint_{\Phi_{i_0}^0} f_{\theta_{i_0}}(z+\sigma)  \text{d} z,
\end{align}
where we have used the fact that  $f_{\theta_{i_0}-\sigma}(z)=f_{\theta_{i_0}}(z+\sigma)$. Then, one infers from (\ref{prx}), (\ref{pry}), (\ref{prxy})  and (\ref{prxyz})  that
\begin{align*}
\Pr\{\mathcal {A}(x)\in &\mathcal {O}\}=\prod_{l=1}^n \Pr\{\mathcal {A}(x_l)\in \mathcal {O}_l\} \\&=\Pr\{\mathcal {A}(x_{i_0})\in \mathcal {O}_{i_0}\} \prod_{l=1, l\neq i_0}^n \Pr\{\mathcal {A}(x_l)\in \mathcal {O}_l\} \\&\leq c_b\Pr\{\mathcal {A}(y_{i_0})\in \mathcal {O}_{i_0}\} \prod_{l=1, l\neq i_0}^n \Pr\{\mathcal {A}(y_l)\in \mathcal {O}_l\} \\&~~+\oint_{\Phi_{i_0}^0} f_{\theta_{i_0}}(z+\sigma)  \text{d} z \prod_{l=1, l\neq i_0}^n \Pr\{\mathcal {A}(y_l)\in \mathcal {O}_l\} \\&\leq c_b \Pr\{\mathcal {A}(y)\in \mathcal {O}\}+\max_{i\in V}\oint_{\Phi_{i}^0} f_{\theta_{i}}(z+\sigma)  \text{d} z,
\end{align*}
which means that $\mathcal {A}$  is $(\epsilon, \delta)$-differentially private.

Next, if $\mu(\cup_{i\in V}\Phi_i^0)=0$, we have that $\mu(\Phi_{i}^0)=0$ holds for $\forall i\in V$. Then, we obtain that
\[\delta=\max_{i\in V}\oint_{\Phi_{i}^0} f_{\theta_{i}}(z+\sigma)  \text{d} z=0,\]
i.e.,
$\mathcal {A}$  is $\epsilon$-differentially private.
\end{proof}
\end{theorem}

In the above theorem, note that
\begin{align*}
 &\oint_{\Phi_{i}^0}f_{\theta_{i}}(z+\sigma)   \text{d} z
\leq \oint_{\Phi_{i}^0\cap \{\mathbf{R}-\{\Phi_{i}^0+\sigma\}\}}f_{\theta_{i}}(z+\sigma)   \text{d} z
\\&\leq \max_{\Phi\subset \mathbf{R}, \mu(\Phi)=\mu(\Phi_{i}^0\cap \{\mathbf{R}-\{\Phi_{i}^0+\sigma\}\})}\oint_{ \Phi} f_{\theta_{i}}(z) \text{d} z.
\end{align*}
It means that $\delta$ satisfies
\begin{align}\label{worsdelt}
\delta &\leq \max_{\Phi\subset \mathbf{R}, \mu(\Phi)=\mu(\Phi_{i}^0\cap \{\mathbf{R}-\{\Phi_{i}^0+\sigma\}\}), i\in V}\oint_{ \Phi} f_{\theta_{i}}(z) \text{d} z\nonumber\\&\leq \max_{i\in V}\left[\mu(\Phi_{i}^0\cap \{\mathbf{R}-\{\Phi_{i}^0+\sigma\}\}) \sup_{z\in \mathbf{R}} f_{\theta_{i}}(z)\right].
\end{align}
Meanwhile, we have
\begin{align*}
\lim_{\sigma\rightarrow 0}\mu(\Phi_{i}^0\cap \{\mathbf{R}-\{\Phi_{i}^0+\sigma\}\})=0.
\end{align*}
Hence, when the PDF of the adding noise is given, we have
\begin{align*}
\lim_{\sigma\rightarrow 0}\delta=0,
\end{align*}
i.e., smaller $\sigma$-adjacency vectors can guarantee a smaller $\delta$ for  $(\epsilon, \delta)$-differential privacy.

From Theorem \ref{theorem1:scp}, it is known that  (\ref{factpingyi0}) is a sufficient condition of $(\epsilon, \delta)$-differential privacy. However, it should be pointed out that (\ref{factpingyi0}) is not a necessary condition of $(\epsilon, \delta)$-differential privacy (though it is a necessary condition of $\epsilon$-differential privacy). An example is Gaussian noise, which is $(\epsilon, \delta)$-differentially private noise \cite{TIT16, Dwork08}, but  (\ref{factpingyi0})  no longer holds for Gaussian noise.  The detailed analysis will be given in the next subsection.
Then, we give the other useful sufficient condition of $(\epsilon, \delta)$-differential privacy as follows, which can be used to prove that Gaussian noise is $(\epsilon, \delta)$-differential privacy.

\begin{theorem}\label{theorem1:scpad}
Let $\Theta=\Theta_0\cup \Theta_1$. Assume that
\begin{equation}\label{con3}
\oint_{\Theta_0} f_{\theta_i} (z) \text{d} z \leq \delta, \forall i\in V
\end{equation}
 and (\ref{factpingyi0}) holds when $\theta\in \Theta_1$, i.e.,
\begin{equation}\label{con4}
\sup_{\hat{\sigma}\in[-\sigma, \sigma], \theta\in \Theta_1} {f_{\theta_i+\hat{\sigma}}(z)\over f_{\theta_i}(z)} \leq c_b.
\end{equation}
Then, $\mathcal {A}$ is $(\epsilon, \delta)$-differentially private, where  $\epsilon\leq \log(c_b)$.
\end{theorem}
\begin{proof}
Given any $\sigma$-adjacent state vectors $x$ and $y$ satisfying $x_{i_0}=y_{i_0}-\sigma$ and $x_{i}=y_{i}$ (when $i\neq i_0$), we have
\begin{align*}
&\Pr\{\mathcal {A}(x)\in \mathcal {O}\}=\prod_{l=1}^n \Pr\{\mathcal {A}(x_l)\in \mathcal {O}_l\} \\&=\Pr\{\mathcal {A}(x_{i_0})\in \mathcal {O}_{i_0}\} \prod_{l=1, l\neq i_0}^n \Pr\{\mathcal {A}(x_l)\in \mathcal {O}_l\} \\&=\left[\Pr\{\mathcal {A}(x_{i_0})\in \mathcal {O}_{i_0}|\theta\in \Theta_0\} +\Pr\{\mathcal {A}(x_{i_0})\in \mathcal {O}_{i_0}| \theta\in \Theta_1\} \right]\\ &~~~\times \prod_{l=1, l\neq i_0}^n \Pr\{\mathcal {A}(y_l)\in \mathcal {O}_l\}\\&\leq c_b\Pr\{\mathcal {A}(y_{i_0})\in \mathcal {O}_{i_0}\} \prod_{l=1, l\neq i_0}^n \Pr\{\mathcal {A}(y_l)\in \mathcal {O}_l\}\\ &~~~+ \oint_{\Theta_0} f_{\theta_i} (z) \text{d} z\prod_{l=1, l\neq i_0}^n \Pr\{\mathcal {A}(y_l)\in \mathcal {O}_l\}\\&\leq c_b\Pr\{\mathcal {A}(y)\in \mathcal {O}\}+\delta.
\end{align*}
Thus, we have completed the proof.
\end{proof}

Considering the conditions in Theorem \ref{theorem1:scpad},  for any kinds of noise random distribution, we have
\[\lim_{\mu(\Theta_0)\rightarrow \mu(\Theta)}\oint_{\Theta_0} f_{\theta_i} (z) \text{d} z =1,\]
and
\begin{align*}
\lim_{\mu(\Theta_1)\rightarrow 0}\sup_{\hat{\sigma}\in[-\sigma, \sigma], \theta\in \Theta_1} {f_{\theta_i+\hat{\sigma}}(z)\over f_{\theta_i}(z)} =1.
\end{align*}
Hence, it follows from Theorem \ref{theorem1:scpad} that using any kinds of random noise, $\mathcal {A}$  is $(0, 1)$-differentially private. This can also be followed from the truth that
\[\Pr\{\mathcal {A}(x)\in \mathcal {O}\}-\Pr\{\mathcal {A}(y)\in \mathcal {O}\}\leq 1\]
holds for any kinds of noise adding mechanism (because $0\leq\Pr\{\cdot\}\leq 1$ always holds true). Thus, it is meaningless to consider a $(0, 1)$-differentially private mechanism, since it can be satisfied by any random distributions.
Note that if $\Theta=\Theta_0(k)\cup \Theta_1(k)$ and $\Theta_0(k)\subset \Theta_0(k+1)$, where $\Theta_1(\infty)=\Theta$, then we have
\begin{align*}
\oint_{\Theta_0(k)} f_{\theta_i} (z) \text{d} z & \leq \oint_{\Theta_0(k+1)} f_{\theta_i} (z) \text{d} z \\ &  \leq \oint_{\Theta_0(\infty)} f_{\theta_i} (z) \text{d} z=1
\end{align*}
while
\begin{align*}
\sup_{\hat{\sigma}\in[-\sigma, \sigma], \theta\in \Theta_1(k)} {f_{\theta_i+\hat{\sigma}}(z)\over f_{\theta_i}(z)} &\geq \sup_{\hat{\sigma}\in[-\sigma, \sigma], \theta\in \Theta_1(k+1)} {f_{\theta_i+\hat{\sigma}}(z)\over f_{\theta_i}(z)}\\&\geq \sup_{\hat{\sigma}\in[-\sigma, \sigma], \theta\in \Theta_1(\infty)} {f_{\theta_i+\hat{\sigma}}(z)\over f_{\theta_i}(z)}=1,
\end{align*}
where we have used the fact that $\Theta_1(\infty)=\emptyset$.
It means that there exists an increasing sequence $\delta(k)$ and an decreasing sequence $\epsilon(k)=\log(c(k))$ satisfying $\lim_{k\rightarrow \infty}\delta(k)=1$ and $\lim_{k\rightarrow \infty}\epsilon(k)=0$, respectively, such that  $(\epsilon(k), \delta(k))$-differential privacy is guaranteed by $\mathcal{A}$.  However, it should be pointed out that different noise distribution can guarantee the different smallest $\delta$ and different corresponding $\epsilon$ of $(\epsilon, \delta)$-differential privacy. In Theorem \ref{theorem1:scp}, the estimation of the upper bounds for  $\delta$ and $\epsilon$ can be tighten for some special distributions (e.g., uniform distribution), which will be illustrated in the following subsection.

\subsection{Case Studies}
From the above two theorems, it is not difficult to determine whether the added noise can guarantee the differential privacy of a random mechanism or not. In the following, we analyze differential privacy of some common random noises, including Laplacian, Gaussian, and uniform distribution noise.

First, we consider the Laplacian noise adding mechanism. Assume that the PDF is $f_{\theta_i}(z)={1\over 2b} e^{-{|z-a|\over b}}$, where $a$ and $b$ are two constants.
\begin{figure}[ht]
\begin{center}
\includegraphics[width=0.35\textwidth]{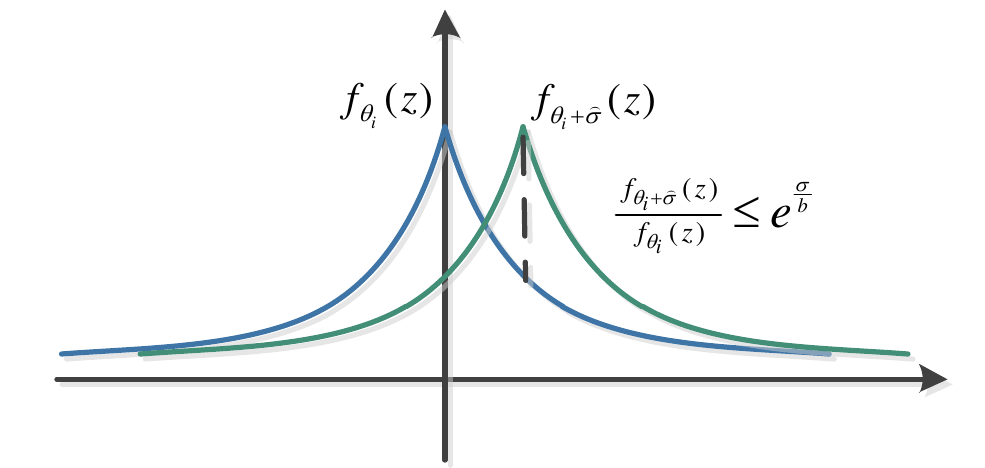}
\vspace{-5pt}
\caption{Laplacian noise: $\epsilon$-differentially private.} \label{sspdf}
\end{center}
\end{figure}
We check the conditions $c_1$ and $c_2$, respectively. From Fig. \ref{sspdf}, it is clear that $c_1$ holds true due to the continuous and non-zero of the PDF of Laplacian noise. Then, note that for $\forall \hat{\sigma}\in [-\sigma, \sigma]$, we have
\begin{align*}
 |{f_{\theta_i+\hat{\sigma}}(z)\over f_{\theta_i}(z)}| &={{1\over 2b} e^{-{|z-\hat{\sigma}-a|\over b}}\over {1\over 2b} e^{-{|z-a|\over b}}}
\\&=e^{{|z-a|-|z-\hat{\sigma}-a|\over b}}  \leq e^{|\sigma|\over b}.
\end{align*}
It means that $c_2$ condition also holds true.
Hence, from Theorem \ref{theorem0:nsc}, it follows that Laplacian noise is an $\epsilon$-differentially private noise.

Then, we consider Gaussian noise. Assume that the PDF of the noise is $f_{\theta_i}(z)={1\over b\sqrt{2 \pi}} e^{-{(z-a)^2\over 2 b^2}}$. Similarly, one infers that $c_1$ holds true for Gaussian noise. Note that
\begin{align*}
 |{f_{\theta_i+\hat{\sigma}}(z)\over f_{\theta_i}(z)}| &={{1\over b\sqrt{2 \pi}} e^{-{(z-\hat{\sigma}-a)^2\over 2 b^2}} \over  {1\over b\sqrt{2 \pi}} e^{-{(z-a)^2\over 2 b^2}}}
=e^{{(z-a)^2-(z-\hat{\sigma}-a)^2\over 2b^2}} \\&=e^{\hat{\sigma} (2z-\hat{\sigma}-2a)\over 2b^2},
\end{align*}
which means that
\begin{align*}
\sup_{\hat{\sigma}\in[-\sigma, \sigma], f_{\theta_i}(z)\neq 0} {f_{\theta_i+\hat{\sigma}}(z)\over f_{\theta_i}(z)}\geq \lim_{z\rightarrow \infty} e^{\hat{\sigma} (2z-\hat{\sigma}-2a)\over 2b^2}=\infty.
\end{align*}
Hence, from Theorem \ref{theorem0:nsc}, it follows that Gaussian noise is not an $\epsilon$-differentially private noise.
\begin{figure}[ht]
\begin{center}
\includegraphics[width=0.35\textwidth]{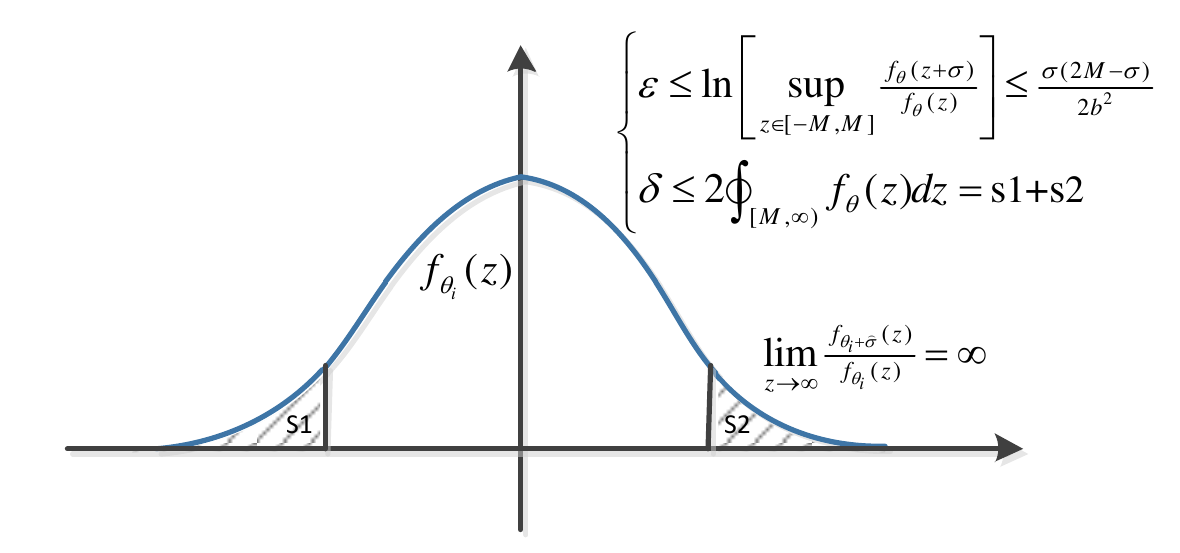}
\vspace{-5pt}
\caption{Gaussian noise: $(\epsilon, \delta)$-differentially private.} \label{ssgspdf}
\end{center}
\end{figure}
However, as shown in Fig. \ref{ssgspdf}, there exists a large constant $M$ such that $\epsilon$ is bounded by
\[\epsilon\leq \ln (\max e^{\hat{\sigma} (2z-\hat{\sigma}-2a)\over 2b^2})\leq  {\sigma (2M-\sigma)\over 2b^2},\]
for $z\in [-M, M]$, and $\delta$ is bounded by
\begin{align*}
\delta&\leq \oint_{(-\infty, -M]\cup [M, \infty)} f_{\theta_i} (z) \text{d} z\\&={1\over b\sqrt{2 \pi}}  \oint_{(-\infty, -M]\cup [M, \infty)} e^{-{(z-a)^2\over 2 b^2}} \text{d} z,
\end{align*}
which is a small value. It means that the conditions in Theorem \ref{theorem1:scpad} can be satisfied. Thus, we infer that Gaussian noise is an  $(\epsilon, \delta)$-differentially private noise.

Lastly, consider the uniform distribution noise with its PDF as ${1\over b-a}$. Clearly, $c_1$  is not true due to the infinite measure of the zero-point set. Hence,  uniform distribution is not an $\epsilon$-differentially private noise. Then, we check the conditions in Theorem \ref{theorem1:scp}. As shown in Fig. \ref{ssunpdf}, it is found that for an uniform distribution noise
\begin{align*}
\sup_{\hat{\sigma}\in[-\sigma, \sigma], f_{\theta_i}(z)\neq 0} {f_{\theta_i+\hat{\sigma}}(z)\over f_{\theta_i}(z)}={{1\over b-a} \over {1\over b-a}}=1
\end{align*}
and
\[\max_{i\in V} \oint_{\Phi_{i}^0} f_{\theta_{i}}(z+\sigma)  \text{d} z = {\sigma \over b-a}.\]
It means that the upper bounds of both $\epsilon$ and $\delta$  in Theorem \ref{theorem1:scp} are tight.
Thus, one infers that uniform noise is an  $(\epsilon, \delta)$-differentially private, where $\epsilon=0$ and $\delta={\sigma \over b-a}$.
\begin{figure}[ht]
\begin{center}
\includegraphics[width=0.35\textwidth]{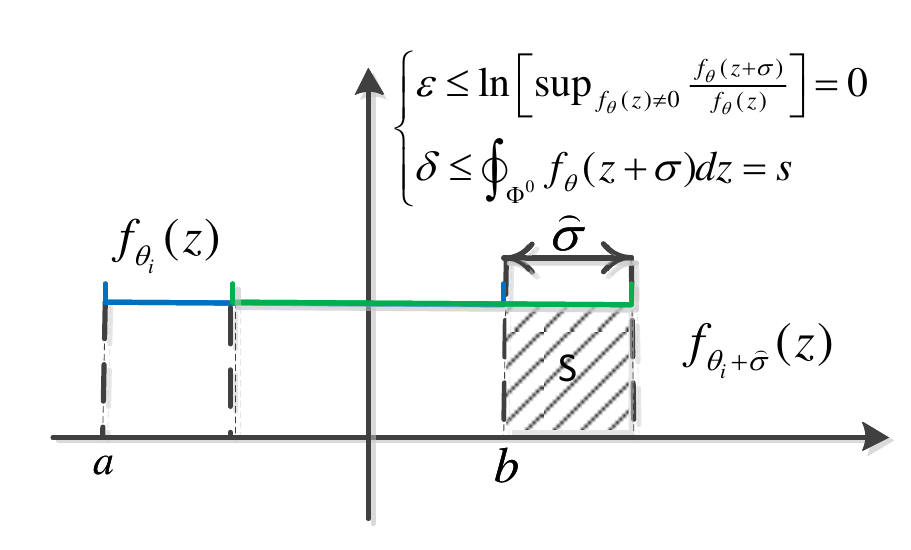}
\vspace{-5pt}
\caption{Uniform noise: $(0, \delta)$-differentially private.} \label{ssunpdf}
\end{center}
\end{figure}
 Then, it is noted that $\delta$ is a decreasing function of $b-a$ and satisfies
 \[\lim_{b-a\rightarrow \infty}\delta=0.\]
 Hence, for any small $\delta$, we can find a corresponding $(0, \delta)$-differentially private uniform noise.
 But, it should be pointed out that when the value of $b-a$ increases, the variance of the uniform distribution ($={(b-a)^2 \over 12}$) increases.

The above analysis shows how to use the developed theory to determine differential privacy of several well-known noise distribution. Using our theory can obtain the same results for a few special cases as those proved in the existing work, which verifies the effectiveness of the proposed theory. In addition, different from the previous approach, our method can determine the differential privacy of the randomized mechanism with any fixed distribution of noise by checking the conditions, and thus is an efficient criterion of differential privacy. For example, we consider the following noise distribution,
\begin{equation*}
f_{\theta_i}(z)=\left\{ \begin{aligned}
       & {z\over 2} e^{-z},  && z \geq0;\\
       & {-z\over 2} e^{z},  &&  z\leq 0.
                          \end{aligned} \right.
                          \end{equation*}
From the definition, it is not obvious to analyze its differential privacy directly. But, from Theorem \ref{theorem0:nsc1}, we easily infer that it is not $\epsilon$-differentially private, since $f_{\theta_i}(0)=0$, and thus it does not satisfy the necessary condition given in the theory.

\section{Application on privacy-preserving consensus algorithm}\label{sec:ap}
Consensus is an efficient algorithm for distributed computing and control, which refers to  the action that nodes in the network reach a global agreement regarding a certain opinion using their local neighbors' information only  \cite{olfati2007consensus}\@.  Consensus has been applied in a variety of areas, e.g., distributed energy management \cite{zhao2015tsg}, scheduling\cite{hetsp15},  and clock synchronization  \cite{schenato2011average, he2011times, carli14tac}. Recently, the privacy-preserving consensus problem has been studied, which aims to guarantee that the privacy of initial state is preserved and at the same time a consensus can still be achieved \cite{ny14tac, Nozari16, yilin15tac}. The basic idea is to add random noises to the real state value during the communication for privacy preserving, the same as (\ref{randoma0}).  This motivates us to adopt the developed theories in the above section to analyze differential privacy of the privacy-preserving consensus algorithm.

\subsection{Privacy-preserving Consensus Algorithm}
A network is abstracted by an undirected and connected graph, $G = (V, E)$, where $V$ is the set of nodes and $E$ is the set of the communication links (edges) between nodes. An edge $(i, j)\in E$ iff nodes $i$ and $j$ can communicate with each other.  Let $N_i$ be the neighbor set of node $i$, defined by $N_i=\{j| j\in V, (i, j)\in E, j\neq i\}$. Let $|V|=n\geq 3$ be the number of nodes in the network and $x_i(0)\in \mathbf{R}$ be the initial state of node $i$. Let $x(0)=[x_1(0), ...., x_n(0)]^T \in \Omega_x^0\subseteq \mathbf{R}^n$.

\textbf{General Consensus Algorithm:}  For a general consensus algorithm, each node will communicate with its neighbor nodes and update its state based on the received information to obtain a consensus of all initial state values. The dynamic iteration equation is given by
\begin{align}\label{generalconsensus}
& x_i(k+1)=w_{ii}x_i(k)+\sum_{j\in N_i} w_{ij} x_j(k),
\end{align}
for $\forall i\in V$, which can be written in the matrix form as
\begin{align}\label{matrix_generalconsensus}
& x(k+1)=Wx(k),
\end{align}
where $w_{ii}$ and $w_{ij}$ are  weights, and $W$ is the weight matrix. It is well known from \cite{Olshevsky11} that, if,
i) $w_{ii}>0$ and $w_{ij}>0$; and ii)  $W$ is a doubly stochastic matrix,
then an average consensus can be achieved by (\ref{matrix_generalconsensus}), i.e.,
\begin{align}\label{gcconvergence}
& \lim_{k\rightarrow \infty}x(k)= {\sum_{\ell=1}^n x_\ell(0)\over n} \mathbf{1}=\bar{x}.
\end{align}

\textbf{Privacy-preserving Consensus (PC) Algorithm:} When the privacy of nodes' initial states are concerned, each node may be unwilling to release its real state to the neighbor nodes at each iteration.  To preserve the privacy of  nodes' initial states,  a widely used approach is to add a random noise to the real state when a node needs to communicate with its neighbor nodes  \cite{yilin15tac}. Hence, we introduce a common privacy-preserving consensus algorithm as follows:
\begin{align} \label{pca}
\mathcal{PC}:~ \left\{ \begin{aligned}
        &x^+(k) =x(k)+\theta(k) \\
          & x(k+1) = W x^+(k)
                          \end{aligned} \right.
                          \end{align}
A privacy-preserving average consensus algorithm is to design the adding noise process (including the noise distribution and the correlations among noises in different iterations),  such that the goal of (\ref{gcconvergence}) is achieved under (\ref{pca}).

\subsection{Privacy Conditions of  Consensus}

We define the input and the output sequences of each node $i$ in privacy-preserving consensus algorithm  (\ref{pca}) until iteration $k$ as
\begin{equation}
\mathcal{I}_{x_i}^{in}(k)=\{x_i(0), \theta_i(0), ..., \theta_i(k)\},
\end{equation}
and
\begin{equation}
\mathcal{I}_{x_i}^{out}(k)=\{x_i^+(0), ...,  x_i^+(k)\},
\end{equation}
respectively. Then, $\mathcal{I}_x^{in}(k)=\{x(0), \theta(0), ..., \theta(k)\}$ is the system input and $\mathcal{I}_x^{out}(k)=\{x^+(0), ..., x^+(k)\}$ is the system output. Let the information set of the adding noises  for node $i$ be $\mathcal{I}_{i_{noise}}^{in}(k)=\{\theta_i(0), ..., \theta_i(k)\}$. Let $f_{\theta_i(k)}(z)$  be the PDF of $\theta_i(k)$. Then, we have $\textrm{Range}(\mathcal{PC})=\Omega_x^0\oplus \Theta(0)\oplus ... \oplus\Theta(k) \oplus ...$, where $\oplus$ denotes the plus of the elements in two sets.

By referring to \cite{Nozari16}, we introduce the definition of $(\epsilon, \delta)$-differential privacy for a consensus algorithm as follows.
\begin{definition}
A PC algorithm (\ref{pca}) is $(\epsilon, \delta)$-differentially private if, for any pair $x$ and $y$ of $\sigma$-adjacent initial state vector and any set $\mathcal {O}\subseteq \mathbf{R}^{n\times \infty}$,
 \begin{equation}\label{diffprivacy}
\Pr\{\mathcal{I}_{x}^{out}(\infty)  \in \mathcal{O}\}\leq e^{\epsilon}\Pr\{\mathcal{I}_{y}^{out}(\infty) \in \mathcal{O}\}+\delta.
\end{equation}
If $\delta=0$, we say that (\ref{pca}) is $\epsilon$-differentially private.
 \end{definition}

First, we give the necessary condition of $\epsilon$-differential privacy for  algorithm (\ref{pca}).
 \begin{theorem}\label{theorem3:nsp}
If algorithm (\ref{pca}) is $\epsilon$-differentially private, then for $\forall k\geq0$, the random noise vector $\sum_{l=0}^{k} W^{k-l} \theta(l)$ should satisfy conditions $c1$ and $c2$.
 \begin{proof}
  Let $\mathcal{O}^{n\times k}\subseteq \mathbf{R}^{n\times \mathbf{k}}$ for $k>0$ and $\mathcal{O}^{n\times 0}\subseteq \mathbf{R}^{n}$ for $k=0$. For any pair $x$ and $y$ of $\sigma$-adjacent initial state vector, we have
\begin{align}\label{addn}
&\Pr\{\mathcal{I}_{x}^{out}(\infty)  \in \mathcal{O}\}\leq e^{\epsilon}\Pr\{\mathcal{I}_{y}^{out}(\infty) \in \mathcal{O}\}, \forall \mathcal {O}\subseteq \mathbf{R}^{n\times \infty}
\nonumber\\\Leftrightarrow &\Pr\{\mathcal{I}_{x}^{out}(k)  \in \mathcal{O}^{n\times k}\}\leq e^{\epsilon}\Pr\{\mathcal{I}_{y}^{out}(k) \in \mathcal{O}^{n\times k}\},  \\& \forall k\geq0, \mathcal{O}^{n\times k}\subseteq \mathbf{R}^{n\times k} \nonumber\\
\Rightarrow & \Pr\{x^+(k)  \in \mathcal{O}^{n\times 1}\}\leq e^{\epsilon}\Pr\{y^+(k)  \in \mathcal{O}^{n\times 1}\}, \\& \forall k\geq0, \mathcal{O}^{n\times k}\subseteq \mathbf{R}^{n\times k}. \nonumber
\end{align}

From (\ref{pca}), we have
\begin{align}\label{xkwt}
x^+(k) &=x(k)+\theta(k)\nonumber\\
&= W[x(k-1)+\theta(k-1)]+\theta(k)
\nonumber\\&= W^kx(0)+ \sum_{l=0}^{k} W^{k-l} \theta(l)
\nonumber\\&= x(0)+(W^k-I)x(0)+ \sum_{l=0}^{k} W^{k-l} \theta(l),
\end{align}
where $I$ is an identity matrix. From (\ref{addn}), (\ref{xkwt}) and Theorem \ref{theorem0:nsc}, we infer that $(W^k-I)z+\sum_{l=0}^{k} W^{k-l} \theta(l), z=x, y$  should satisfy conditions $c1$ and $c2$ for any  $\sigma$-adjacent state vector $x$ and $y$. It follows that
$\sum_{l=0}^{k} W^{k-l} \theta(l)$ satisfies conditions $c1$ and $c2$.
 \end{proof}
 \end{theorem}

Then, we give the necessary condition of average consensus for algorithm (\ref{pca}).
  \begin{theorem}\label{theorem4:cn}
Using algorithm (\ref{pca}), if
\begin{align}\label{covp}
\Pr\{\lim_{k\rightarrow \infty}x(k)=\bar{x}\}=1,
\end{align}
i.e., the average consensus is achieved in probability, then
\[\Pr\{\lim_{k\rightarrow \infty} \sum_{l=0}^{k-1} W^{k-l} \theta(l)=0\}=1,\]
and $\Pr\{\lim_{k\rightarrow \infty}W\theta(k)=0\}=1$,
i.e., the add noise should equal $0$ or be the $0$-eigenvector of $W$ when $k\rightarrow \infty$.
\begin{proof}
Under algorithm (\ref{pca}), we have
\begin{align*}
\lim_{k\rightarrow \infty}x(k)=&\lim_{k\rightarrow \infty}\left[W^{k}x(0)+ \sum_{l=0}^{k-1} W^{k-l} \theta(l)\right]\\
=&\lim_{k\rightarrow \infty} W^kx(0)+\lim_{k\rightarrow \infty}   \sum_{l=0}^{k-1} W^{k-l} \theta(l)\\
=&\bar{x}+\lim_{k\rightarrow \infty}  \sum_{l=0}^{k-1} W^{k-l} \theta(l),
\end{align*}
where set $\sum_{l=1}^{-1}(\cdot)=0$.
Then, from (\ref{covp}), it follows that
\begin{align*}
\Pr\{\lim_{k\rightarrow \infty}  \sum_{l=0}^{k-1} W^{k-l} \theta(l)=0\} &=\Pr\{\lim_{k\rightarrow \infty}[x(k)-\bar{x}]=0\}\\&=1.
\end{align*}
Then, note that when $\sum_{l=0}^{k-1} W^{k-l} \theta(l)=0$, we have $W\theta(\infty)=0$. Hence, we have $\Pr\{\lim_{k\rightarrow \infty}W\theta(k)=0\}=1$.
\end{proof}
 \end{theorem}

Next, by comparing the necessary conditions of $\epsilon$-differential privacy and average consensus, an impossibility result is given as follows.

\textbf{Impossibility Result:} From Theorem \ref{theorem4:cn}, one infers that the added noise $\theta(k)$ should converge to $0$ or the $0$-eigenvector of $W$, denoted by $\lambda_0$, i.e., $\lim_{k\rightarrow \infty}\theta(k)=0$ or $\lim_{k\rightarrow \infty}\theta(k)=\lambda_0$.
Note that
\begin{align*}
 \lim_{k\rightarrow \infty}\sum_{l=0}^{k} W^{k-l} \theta(l) =& \lim_{k\rightarrow \infty}\left[\sum_{l=0}^{k-1} W^{k-l} \theta(l)+\theta(k)\right].
\end{align*}
Then, we have
\begin{align*}
 \Pr\{\lim_{k\rightarrow \infty}\sum_{l=0}^{k} W^{k-l} \theta(l)=0 ~ \textrm{or}~ \lambda_0\} =1.
\end{align*}
Thus, the conditions $c1$ and $c2$ no longer hold for the added noise $\sum_{l=0}^{k} W^{k-l} \theta(l)$ when $k\rightarrow \infty$. It  contradicts to the necessary condition in Theorem \ref{theorem3:nsp}, and thus $\epsilon$-differential privacy cannot be guaranteed. It means that the necessary condition of differential privacy and the necessary condition of average consensus are conflicted, which leads to the impossibility result. Hence, using (\ref{pca}), nodes cannot simultaneously converge to the average of their initial states and preserve $\epsilon$-differential privacy of their initial states.


Finally, we provide the sufficient conditions of differential privacy for  algorithm (\ref{pca}).
\begin{theorem}\label{theorem:scac}
 Suppose that the noise sequence $\theta(1), \theta(2),..., \theta(k), ....$, is independent from both $\theta(0)$ and $x(0)$. Then, if  $\theta(0)$ satisfies conditions $c1$ and $c2$,  algorithm (\ref{pca}) provides $\epsilon$-differential privacy;  if  $\theta(0)$ satisfies (\ref{factpingyi0}) or (both (\ref{con3}) and  (\ref{con4}) simultaneously), algorithm (\ref{pca}) provides $(\epsilon, \delta)$-differential privacy, where $ \epsilon\leq \log(c_b)$ and $\delta$ satisfies (\ref{worsdelt}) or (is given by  (\ref{con3})).
\begin{proof}
Given any $\mathcal {O}\subseteq \mathbf{R}^{n\times \infty}$, we let $\mathcal {O}_l^{\iota}$ be the set of the $l$-th to $\iota$-th column vectors of $\mathcal {O}$ for $l, \iota\in \mathbf{N}^+$. Then,
 \begin{align*}
&\Pr\{\mathcal{I}_{x}^{out}(\infty)  \in \mathcal{O}\}\\=&\Pr\{x^+(0) \in \mathcal {O}_1^{1}\} \Pr\{ \mathcal{I}_{x}^{out}(1, \infty)\in  \mathcal {O}_2^{\infty}|x^+(0)\in\mathcal {O}_1^{1}\}\\=& \oint_{ \mathcal {O}_1^{1}} f_{\theta(0)}(z-x)\Pr\{ \mathcal{I}_{x}^{out}(1, \infty)\in  \mathcal {O}_2^{\infty}|z\}  \text{d} z
\end{align*}
and
 \begin{align*}
&\Pr\{\mathcal{I}_{y}^{out}(\infty)  \in \mathcal{O}\}
\\=&\Pr\{y^+(0) \in \mathcal {O}_1^{1}\} \Pr\{ \mathcal{I}_{y}^{out}(1, \infty)\in  \mathcal {O}_2^{\infty}|y^+(0)\in\mathcal {O}_1^{1}\}
\\=& \oint_{ \mathcal {O}_1^{1}} f_{\theta(0)}(z-y)\Pr\{ \mathcal{I}_{y}^{out}(1, \infty)\in  \mathcal {O}_2^{\infty}|z\}  \text{d} z
\end{align*}
Since  $\theta(1), \theta(2),..., \theta(k), ....$, are independent from both $\theta(0)$ and $x(0)$, for given any same vector $z$, we have
 \begin{align*}
 & \Pr\{ \mathcal{I}_{x}^{out}(1, \infty)\in  \mathcal {O}_2^{\infty}|z\} =\Pr\{ \mathcal{I}_{y}^{out}(1, \infty)\in  \mathcal {O}_2^{\infty}|z\}.
\end{align*}

When $\theta(0)$ satisfies conditions $c1$ and $c2$, we have
 \begin{align*}
& \oint_{ \mathcal {O}_1^{1}} f_{\theta(0)}(z-x)\Pr\{ \mathcal{I}_{x}^{out}(1, \infty)\in  \mathcal {O}_2^{\infty}|z\}  \text{d} z\\
=& \oint_{ \mathcal {O}_1^{1}} f_{\theta(0)}(z-y+\mathbf{\sigma})\Pr\{ \mathcal{I}_{x}^{out}(1, \infty)\in  \mathcal {O}_2^{\infty}|z\}  \text{d} z
\\ = & \oint_{ \mathcal {O}_1^{1}} f_{\theta_{i_0}(0)}(z_{i_0}-y_{i_0}+\sigma_{i_0})\prod_{i=1, i\neq i_0}^n f_{\theta_{i}(0)}(z_i-y_i) \\&\times\Pr\{ \mathcal{I}_{x}^{out}(1, \infty)\in  \mathcal {O}_2^{\infty}|z\}  \text{d} z\\ \leq & \oint_{ \mathcal {O}_1^{1}} c_b f_{\theta(0)}(z-y)\Pr\{ \mathcal{I}_{y}^{out}(1, \infty)\in  \mathcal {O}_2^{\infty}|z\}  \text{d} z,
\end{align*}
where $\mathbf{\sigma}\in \mathbf{R}^n$ is a vector with $\mathbf{\sigma}_{i_o}=\sigma$ and all the other elements equal to $0$, which means that
\[\Pr\{\mathcal{I}_{x}^{out}(\infty)  \in \mathcal{O}\}\leq e^{\epsilon}\Pr\{\mathcal{I}_{y}^{out}(\infty) \in \mathcal{O}\}.\]
Thus, (\ref{pca}) provides $\epsilon$-differential privacy.

When $\theta(0)$ satisfies (\ref{factpingyi0}), we have
 \begin{align*}
& \oint_{ \mathcal {O}_1^{1}} f_{\theta(0)}(z-x)\Pr\{ \mathcal{I}_{x}^{out}(1, \infty)\in  \mathcal {O}_2^{\infty}|z\}  \text{d} z
 \\ \leq & \oint_{ \mathcal {O}_1^{1}} c_b f_{\theta(0)}(z-y)\Pr\{ \mathcal{I}_{y}^{out}(1, \infty)\in  \mathcal {O}_2^{\infty}|z\}  \text{d} z\\&+
 \oint_{\hat{\mathcal {O}}_1^{1}}  f_{\theta(0)}(z-x)\Pr\{ \mathcal{I}_{y}^{out}(1, \infty)\in  \mathcal {O}_2^{\infty}|z\}  \text{d} z
 \\ \leq & \oint_{ \mathcal {O}_1^{1}} c_b f_{\theta(0)}(z-y)\Pr\{ \mathcal{I}_{y}^{out}(1, \infty)\in  \mathcal {O}_2^{\infty}|z\}  \text{d} z\\&+
 \oint_{\hat{\mathcal {O}}_1^{1}}  f_{\theta(0)}(z-x) \text{d} z
\end{align*}
where $\hat{\mathcal{O}}_1^{1}=\{z|z\in\mathcal {O}_1^{1},  f_{\theta_{i_0}(0)}(z_{i_0}-y_{i_0})=0, f_{\theta_{i_0}(0)}(z_{i_0}-x_{i_0})\neq 0 \}$. Let $\delta$ satisfy  (\ref{worsdelt}). Then, we have
\[\Pr\{\mathcal{I}_{x}^{out}(\infty)  \in \mathcal{O}\}\leq e^{\epsilon}\Pr\{\mathcal{I}_{y}^{out}(\infty) \in \mathcal{O}\}+\delta.\]
Thus, (\ref{pca}) provides $(\epsilon, \delta)$-differential privacy.

If $\theta(0)$ satisfies both (\ref{con3}) and  (\ref{con4}) simultaneously, then there also exists $\Theta_0$ and $\Theta_1$ such that $\theta(0)+x$ satisfies (\ref{con3}) and  (\ref{con4}). Hence, we have
 \begin{align*}
 &\Pr\{\mathcal{I}_{x}^{out}(\infty)  \in \mathcal{O}\}\\
= &\oint_{ \mathcal {O}_1^{1}} f_{\theta(0)}(z-x)\Pr\{ \mathcal{I}_{x}^{out}(1, \infty)\in  \mathcal {O}_2^{\infty}|z\}  \text{d} z
 \\ =& \oint_{ \mathcal {O}_1^{1}\cap \Theta_0} f_{\theta(0)}(z-x)\Pr\{ \mathcal{I}_{x}^{out}(1, \infty)\in  \mathcal {O}_2^{\infty}|z\}  \text{d} z \\ &+ \oint_{ \mathcal {O}_1^{1}\cap \Theta_1} f_{\theta(0)}(z-x)\Pr\{ \mathcal{I}_{x}^{out}(1, \infty)\in  \mathcal {O}_2^{\infty}|z\}  \text{d} z
 \\ \leq &  \oint_{\Theta_0} f_{\theta(0)}(z-x) \text{d} z\\&+ c_b \oint_{ \mathcal {O}_1^{1}\cap \Theta_1} f_{\theta(0)}(z-y)\Pr\{ \mathcal{I}_{y}^{out}(1, \infty)\in  \mathcal {O}_2^{\infty}|z\}  \text{d} z
  \\ \leq &  \delta + \log(c_b)\Pr\{\mathcal{I}_{y}^{out}(\infty)  \in \mathcal{O}\}.
\end{align*}
It means that (\ref{pca}) provides $(\epsilon, \delta)$-differential privacy.

Thus, we have completes the proof.
\end{proof}
\end{theorem}

Based on the above theoretical results, it is not difficult to analyze differential privacy of the existing privacy-preserving consensus algorithm. For example, in \cite{huang12, Nozari16}, the privacy-preserving consensus algorithms are designed by adding  independent and Laplacian noise  to the consensus process, and thus the sufficient conditions in Theorem \ref{theorem:scac} are satisfied. Hence, these privacy-preserving consensus algorithms proposed in \cite{huang12, Nozari16} are $\epsilon$-differentially private, while the exact average consensus cannot be guaranteed by these algorithms. In \cite{yilin15tac},  the exponentially decaying and zero-sum normal noises are adopted in the privacy-preserving consensus algorithm. Since the sum of all added noises equals $0$, the necessary condition in Theorem \ref{theorem3:nsp} cannot be satisfied. Hence, the algorithm proposed in \cite{yilin15tac} is not $\epsilon$-differentially private. But it achieves the exact average consensus in the mean-square sense.

\section{Related Works}\label{sec:rw}
The concept of differential privacy (including $\epsilon$-differential privacy and $(\epsilon, \delta)$-differential privacy) was first introduced by Dwork et al. \cite{Dwork06, Dwork08}.  After then, differential privacy has attracted substantial attention throughout computer, control and
communication science, including areas like deep learning \cite{Abadi16}, optimization \cite{huang15, Han2014}, dynamic systems \cite{ny14tac} and more. There are also some other privacy definitions, e.g., identifiability and mutual-information privacy, and we refer the readers to  \cite{weinatit16}  to view the relationship among these privacy concepts.

Dwork et al.\cite{Dwork06} showed that the Laplacian mechanism, i.e., adding random noise with Laplace distribution proportional
to the global sensitivity of the query function to perturb the query output, can preserve
 $\epsilon$-differential privacy. Also, McSherry and Talwar  \cite{McSherry07} showed the exponential
mechanism can preserve $\epsilon$-differential privacy for general query functions. It is shown that adding random noise with Gaussian distribution can preserve the $(\epsilon, \delta)$-differential privacy for both real valued
query functions \cite{Dwork08} and infinite dimensional query functions \cite{Hall13}.
For the work on enforcing differential privacy in optimization, linear programs are solved in a framework that allows for keeping objective
functions or constraints private \cite{Hsu2014}. This work was extended by the authors of \cite{Han2014}, and they considered
a similar setting wherein some affine objectives with linearly constrained problems are solved while keeping the privacy of the
objective functions. To keep inputs private from an adversary observing a system's outputs, differential privacy has been adapted to dynamical
systems, which introduces the privacy concerns in the context of systems theory \cite{ny14tac}. Wasserman and Zhou in \cite{Wasserman10}, proposed a statistical framework for differential privacy, where the differential privacy is investigated from a statistical perspective.

Recently, privacy issues are concerned in multi-agent systems, and mainly investigate the privacy-preserving consensus problem.  The objective is aiming to guarantee the agents initial states (or objective functions) private and average consensus to be achieved \cite{ny14tac, Manitara13, huang12,huang15, Nozari16,yilin15tac}. In \cite{huang15}, the authors solved distributed consensus problems while keeping the
agents' objective functions private, and in \cite{huang12} the same authors solved similar
problems while keeping the privacy of each agent's initial state. In these works, differential privacy is guaranteed by adding the independently and Laplacian noises to the consensus process. More recently, Nozari et al. \cite{Nozari16} obtained and
proved an interesting impossibility result that achieving average consensus and differential privacy simultaneously is impossible by contradiction via the definition of differential privacy. This result is also proved in this paper by comparing the necessary conditions of differential privacy and of the average consensus.

Different from the existing work, in this paper, we obtain a necessary and sufficient condition of $\epsilon$-differential privacy, and the sufficient conditions of $(\epsilon, \delta)$-differential privacy, which can be used to  the random noise adding mechanism with any noise distributions.

\section{Conclusions}\label{sec:conclusions}

In this paper, we provided different conditions  of differential privacy for a generally random noise mechanism. We obtained the conditions  for determining differential privacy of random noise mechanism, followed by an application study on privacy-preserving consensus algorithm. Specifically, considering a generally random noise adding mechanism, we obtained a necessary and sufficient condition of $\epsilon$-differential privacy, and two useful sufficient conditions of  $(\epsilon, \delta)$-differential privacy of the noise adding mechanism. We also provided the estimations of the upper bounds of the parameters $\epsilon$ and $\delta$. Then, we showed that the obtained theory provides efficient and simple criteria of differential privacy using case studies. In addition, we applied the obtained result to obtain the necessary condition and the sufficient condition for the privacy-preserving consensus algorithm, under which differential privacy is achieved.

There are still many open issues worth further investigation. For example, in this paper, we focus on the privacy analysis, and do not consider the accuracy of queries from statistical databases under the random noise adding mechanism. How the distribution of the adding noise affect the accuracy of queries needs further investigation. Meanwhile,  the relationship between the parameters in differential privacy ($\epsilon$ and $\delta$), the parameters of the PDF of the adding noise (mean and variance) also needs further investigation. In addition, Cuff and Yu \cite{ccs16} recently introduced a novel privacy concept, named mutual-information differential privacy (MIDP), which is sandwiched between  $\epsilon$-differential privacy and   $(\epsilon, \delta)$-differential privacy in terms of its strength. MIDP depicts differential privacy from the perspective of information theory. Thus, how to extend our results under MIDP is also an interesting issue, requiring further research.

\begin{IEEEbiography}
{Jianping He} (M'15) is currently a research fellow in the Department of Electrical and Computer Engineering at University of Victoria.  He received the Ph.D. degree of  Control Science and Engineering in 2013, at Zhejiang University, Hangzhou, China.  His research interests include the control and optimization of sensor networks and cyber-physical systems, the scheduling and optimization in VANETs and social networks,  and the investment decision in financial market and electricity market. He served as an associate editor for a special issue on Consensus-based Applications in Networked Systems in the International Journal of Robust and Nonlinear Control in 2016.
\end{IEEEbiography}

\begin{IEEEbiography}
{Lin Cai} (S'00-M'06-SM'10)  received the M.A.Sc. and PhD degrees in electrical
and computer engineering from the University of Waterloo, Waterloo,
Canada, in 2002 and 2005, respectively.  Since 2005, she has been with the
Department of Electrical and Computer Engineering, University of Victoria,
where she is currently a Professor. Her research interests span several areas in
communications, networking and control, with a focus on network protocol and
control algorithm design supporting Internet-of-the-Things.

Dr. Cai has served as a TPC symposium co-chair of the IEEE Globecom'10
and Globecom'13. She served as an Associate Editor of the IEEE Transactions on Wireless
Communications, the IEEE Transactions on Vehicular Technology, the EURASIP Journal on
Wireless Communications and Networking, the International Journal of
Sensor Networks, and the Journal of Communications and Networks, and
as the Distinguished Lecturer of the IEEE VTS Society.
She was a recipient of the NSERC Discovery Accelerator Supplement
Grants in 2010 and 2015, respectively, and the Best Paper Awards of the IEEE ICC 2008 and the IEEE
WCNC 2011
\end{IEEEbiography}

\end{document}